\documentclass{article}

\usepackage{microtype}
\usepackage{graphicx}
\usepackage{subfigure}
\usepackage{booktabs} 

\usepackage{bbm}
\usepackage{amsfonts}       
\usepackage{amssymb,amsmath,amsthm}
\usepackage[lined,boxed, linesnumbered]{algorithm2e}
\usepackage{url}

\newcommand{\dg}{\ensuremath{{d}}}
\newcommand{\hdg}{\ensuremath{\hat{d}}}
\newcommand{\hlambda}{\ensuremath{\hat{\lambda}}}

\newcommand{\PolBooksData}{\textsc{PolBooks}\xspace}
\newcommand{\AmazonData}{\textsc{Amazon}\xspace}

\newcommand{\hv}{\ensuremath{\hat{v}}}

\let\oldnl\nl
\newcommand{\nonl}{\renewcommand{\nl}{\let\nl\oldnl}}

\newtheorem{theorem}{Theorem}
\newtheorem{corollary}{Corollary}
\newtheorem{definition}{Definition}
\newtheorem{lemma}{Lemma}
\newtheorem{conjecture}{Conjecture}

\def\DEBUG{true}

\ifdefined\DEBUG

\title{Principal Fairness: \\ Removing Bias via Projections}

\author{%
Aris Anagnostopoulos,
   Luca Becchetti,
Adriano Fazzone\\ 
Cristina Menghini,
Chris Schwiegelshohn\\
  Sapienza University of Rome
}

\begin{document}

\maketitle
	
\begin{abstract}
 Reducing hidden bias in the data and ensuring fairness in algorithmic data analysis has recently received significant attention. We complement several recent papers in this line of research by introducing a general method to reduce bias in the data through random projections in a ``fair'' subspace.
    We apply this method to densest subgraph problem. For densest subgraph, our approach based on fair projections allows to recover both theoretically and empirically an almost optimal, fair, dense subgraph hidden in the input data. We also show that, under the small set expansion hypothesis, approximating this problem beyond a factor of 2 is NP-hard and we show a polynomial time algorithm with a matching approximation bound. 
\end{abstract}

\section*{Acknowledgements}
Partially supported by the ERC Advanced Grant 788893 AMDROMA ``Algorithmic and Mechanism Design Research in Online Markets'' and MIUR PRIN project ALGADIMAR ``Algorithms, Games, and Digital Markets''.

\newcommand{\msdk}{\textsc{MSDk}\xspace}

\section{Introduction}\label{se:intro}

The identification of dense subgraphs  is a fundamental primitive in community
detection and graph mining. Given an underlying graph $G=(V,E)$, the
density of a node set $S\subseteq V$ is defined as
$\frac{2\cdot |E\cap S\times S|}{|S|}$.
In most mining scenarios, communities are assumed to have a high
intra-community density versus a lower inter-community density.
In this sense, density is arguably the most natural measure of quality for
evaluating and comparing communities in graphs (see~\cite{ChakrabortyDMG17} for an extensive survey.)

In this paper, we consider the densest subgraph problem with fairness
constraints. Specifically, we are given a binary labeling of the nodes
of the graph $\ell: V\rightarrow \{-1,1\}$. The labeling corresponds to
an attribute that ideally should be uncorrelated with community
membership.
Our goal is to compute a set of nodes $S\subseteq V$ of maximum density
while ensuring that $S$ contains an equal number of representatives of
either label.
The problem has a number of motivating applications, some of which are 
discussed below.

\paragraph{\bf Mitigation of Polarization}
Social networks are very prone to polarization among
users~\cite{Boxell}:  reinforcement of user preferences
can lead to feedback loops.
For example, recommender systems incentivize disagreement minimization,
leading to echo chambers among users with similar preferences.
This problem has been considered for example by~\cite{MuscoMT18}, who 
studied the problem of identifying a graph of connections between 
users (of two different opinions), 
such that polarization and disagreement  are simultaneously minimized.
The notions behind the fair densest subgraph problem are closely
related: Its goal is to maximize agreement while avoiding
polarization\footnote{The paper by~\cite{MuscoMT18} is
similar in spirit, but very different in terms of problem modelling.}. 

\paragraph{\bf Team Formation}
In crowdsourcing, team formation consists in identifying a set of
workers, whose collective skill set incldes all skills that are 
required for processing some given jobs. Lappas
et al.~\cite{LLT09} proposed subgraph density as a way of modeling the
effectiveness of multiple individuals when working together. 
The potential benefits of team diversity are well documented in organizational psychology \cite{Hor05} studies and also  highlighted by recent work (e.g., see \cite{marcolino2013multi} 
and follow-up work). Diversity in turn can be naturally modeled via fairness constraints.

\paragraph{\bf Diversity in Association Rule Mining}
Sozio and Gionis~\cite{SoG10} study dense subgraphs for
association rule mining: Given a set of tags used to label objects, 
the densest subgraph problem allows to determine additional 
related tags that can be used for a better description of the objects. 
It is  common that the tags that are added are semantically identical to 
those already used. 
We argue that an appropriate labelling of the tags followed by solving
the  fair densest subgraph problem allows recovery of a set of tags
that are not only closely related, but also unique.

\paragraph{\bf Algorithmic Fairness}
As pioneered by Chierichetti et al.~\cite{CKLV17}, there has recently
been considerable work on clustering data sets using the disparity of
impact measure. Conceptually, the aim is to perform data analysis such
that the resulting clustering or classifier does not discriminate based
on some protected attribute. In our case, finding a densest subgraph
such that a protected attribute is not disparately impacted is
equivalent to the definition of the fair densest subgraph problem.

\subsection{Contributions} 
As it turns out (see Section \ref{sec:problems}), the fair densest
subgraph problem is intractable in general, while its
unconstrained counterpart can be solved optimally through network flow \cite{G84}.
Nevertheless, we have some quantifiable results regarding approximation
algorithms in special cases. If the underlying graph itself is fair, we
can show that there exists a $2$-approximation algorithm. We further 
show that, assuming the widely used small set expansion hypothesis \cite{RaghavendraS10}, 
this is the best possible. We also consider the case where the graph itself is not fair and we instead
aim for a proportional representation.
For this, in our opinion more flexible variant of the problem, we
show that the results for fair graphs can be extended.
 
Although this worst-case behavior is discouraging, the possibility of 
effective algorithms is not ruled out on practical instances. To this end, we
identify properties that, if satisfied by some subgraph of the network 
under consideration, will afford recovery of an approximately fair, 
dense subgraph.
More precisely, our goal in this respect is designing a heuristic that
\begin{itemize}
\item[(a)] has a quantifiable guarantee if the underlying graph is
well-behaved and
\item[(b)] is practically viable.
\end{itemize}
Our main result is a spectral algorithm that satisfies both of these requirements.
In particular, the practical viability of our algorithm underscores that
our notion of a well-behaved graph is a realistic one.
As a candidate application, we considered the scenario of providing
diverse recommendations of high quality, using data from the Amazon
product co-purchasing graph.
Our experiments not only confirm the quality of the output solutions,
but also the scalability of our approach, which
may not be the case for a conventional combinatorial approximation
algorithm.

%


\paragraph{Overview of approach.}
Our approach builds on the 
finding \cite{KV99, McS01} that the densest subgraph problem admits a
spectral formulation.
Specifically, an approximate densest subgraph can be computed by 
selecting nodes for inclusion according to the magnitudes of the  
corresponding entries in the main eigenvector of $G$'s adjacency matrix. 
Unfortunately, this approach does not afford balanced solutions in 
general. In a nutshell, we sidestep this issue by first projecting the 
adjacency matrix onto a suitable ``fair'' subspace, an operation that corresponds to the 
enforcement of ``soft'' fairness constraints.

\noindent 
To see why the conventional spectral approach of~\cite{KV99} may not work\footnote{In fact, this applies
to any approach based on unconstrained maximization of the induced 
subgraph's  density.} and why our approach mitigates the issue,
Figure~\ref{fig:polbooksembedding} presents plots obtained from
Amazon
books on US politics \cite{PolBookNetwork}.
The books are labeled as either conservative or liberal, which corresponds
to the labels $-1$ or $1$.
As described above, a candidate application may be to find a selection
of books that are of interest to multiple readers, while mitigating
potential polarization along political lines.

On the left, we observe the books ordered according
to their corresponding entries in the main eigenvector of the adjacency matrix of the
co-purchase graph.
Books are also colored according to political orientation.
We can observe that, whereas liberal books are well distributed,
conservative ones are clustered. On the right we observe the results after
application of our spectral embedding, which affords recovery of a 
subgraph of the co-purchase graph that is both dense and approximately balanced. 
Note that now conservative books are also
well distributed along the principal component.

\begin{figure*}[h]
\begin{center}
\begin{minipage}{0.45\textwidth}
\includegraphics[scale=0.38]{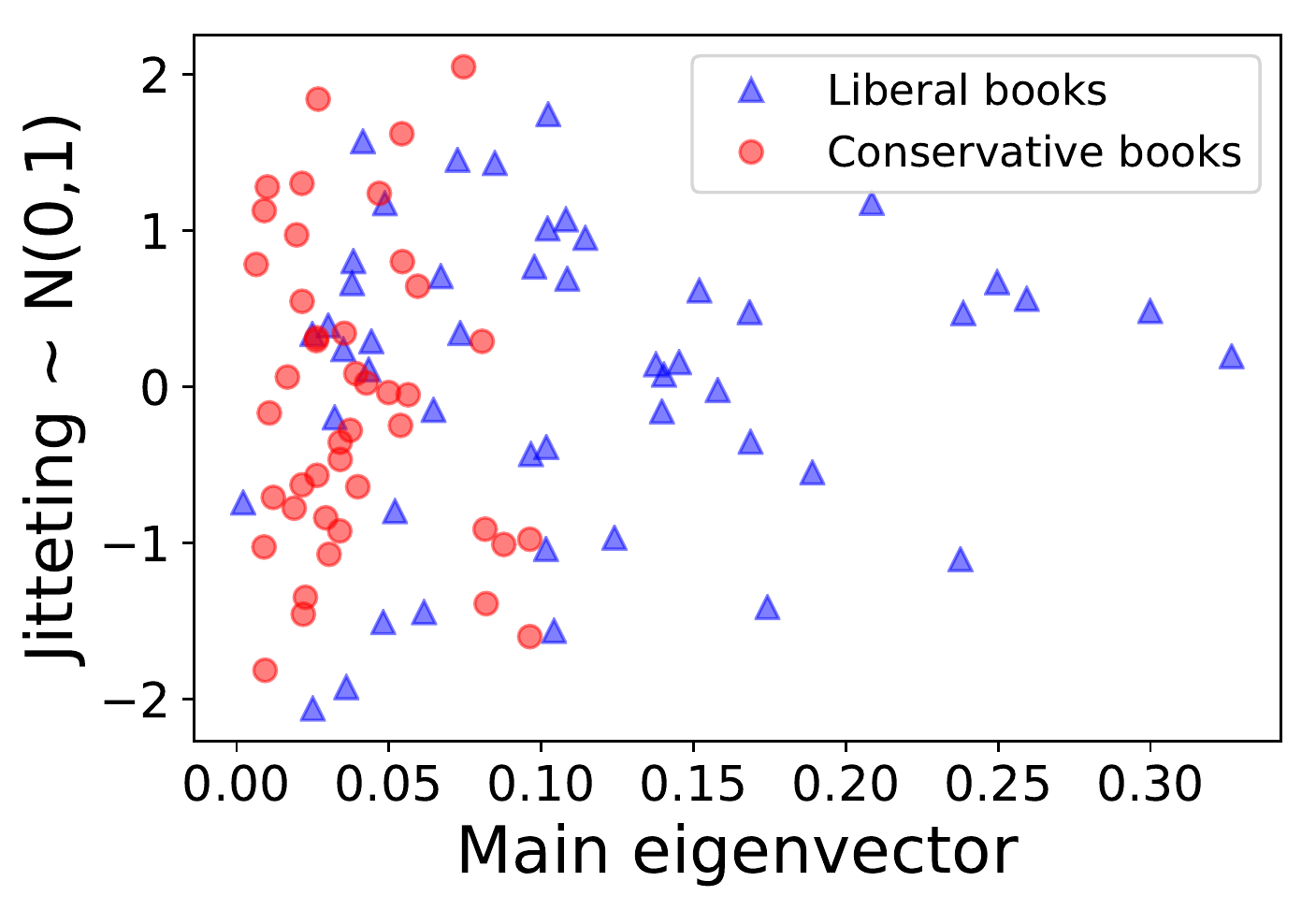}
\end{minipage}
\begin{minipage}{0.45\textwidth}
\includegraphics[scale=0.4]{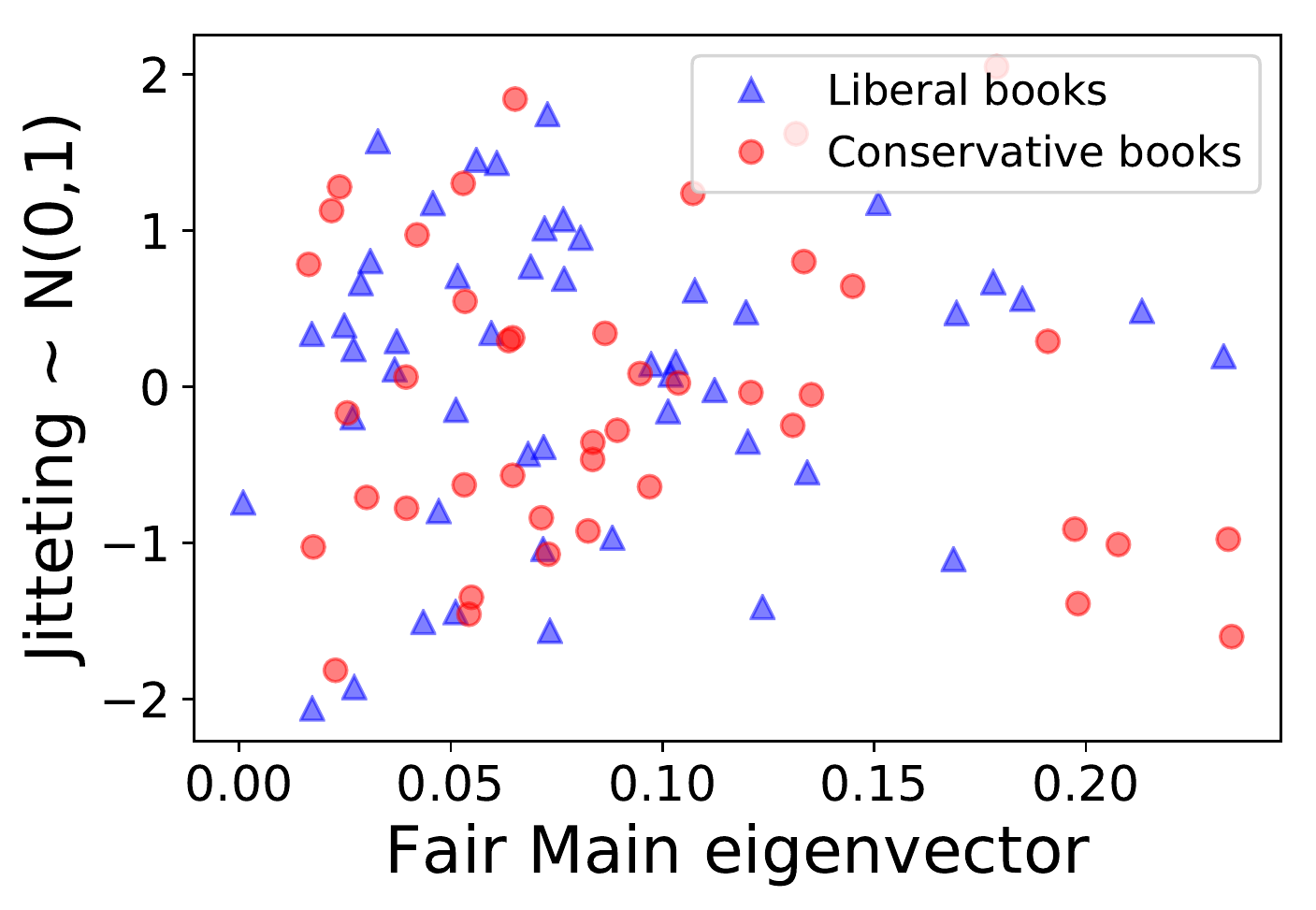}
\end{minipage}
\end{center}
\caption{Projection of books (see Section~\ref{se:experiments}) onto 
the first principal component. (Left) Original data. (Right) Data after 
spectral embedding.
Books are ordered on the $x$ axis according to their corresponding 
entries in the main eigenvector, whereas on the $y$ axis we have random noise
for visualization. 
}
\label{fig:polbooksembedding}
\end{figure*}

\subsection{Related work}
\label{sec:related}
\paragraph{Densest Subgraph}
Identifying dense subgraphs is a key primitive in a number of 
applications; see \cite{fratkin2006motifcut,GibsonKT05,gionis2013piggybacking,TBGGT13}.
The problem can be solved optimally in polynomial
time~\cite{G84}. On the contrary, the fair densest subgraph problem
is highly related to the densest subgraph problem limited to at most $k$
nodes, which cannot be approximated up to a factor of
$n^{1/(\log\log n)^c}$ for some $c>0$ assuming the exponential time
hypothesis~\cite{Manurangsi17} and for which state-of-the-art methods
yield an $O(n^{1/4+\varepsilon})$
approximation~\cite{BhaskaraCCFV10}. 



\paragraph{Algorithmic Fairness.}
Fairness in algorithms received considerable attention in 
the recent past, see \cite{HPNS16,TRT11,ZVG17,ZVGGW17c} and references therein.
The closely related notion of disparate impact was first proposed by~\cite{FFMSV15}. It
has since been used by~\cite{ZVGG17} and Noriega-Campero et
al.~\cite{NBGP18} for classification and Celis et al.~\cite{CHV18,CSV18}
for voting and ranking problems. Another problem that received 
considerable attention is fair clustering. This was first proposed as a 
problem by~\cite{CKLV17} in the case of a binary 
protected attribute. It was then investigated for various 
objectives and more color classes in theirs and subsequent work 
~\cite{BIOSVW19,BCN19,BGKKRSS18,HuangJV19,RS18,SSS19}.

Most closely related to our work are the papers by~\cite{SamadiTMSV18,TantipongpipatS19,KleindessnerSAM19}. The former two papers considers the problem of executing a principal component
analysis in a fair manner. Specifically, given a matrix $A$ where the
rows are colored (e.g. every row corresponds to a man or a woman), they
ask for an algorithm that finds the finds a rank $k$ matrix $A'$ whose
residual error $\|A-A'\|$ is small for both types of rows
simultaneously.
While our method is similarly based on using the principal component in
a fair manner, the difference is that we may be forced to treat the
classes differently, if we aim to uncover a dense subgraph as
illustrated in the example mentioned above and in
Figure~\ref{fig:polbooksembedding}.

The latter paper by~\cite{KleindessnerSAM19} considers
spectral clustering problems such as normalized cut. Like our work, they
project the Laplacian matrix of a graph $G$ onto a suitable ``fair''
subspace, and then run $k$-means on the subspace spanned by  the smallest resulting eigenvectors. Under a fair version of the stochastic block model, they show
that this algorithm recovers planted fair partitions.
Our work continues this idea by applying the technique to the densest
subgraph problem.

\subsection{Preliminaries and Notation}\label{se:prelims}
We consider undirected graphs $G(V,E,w)$, where $V$ is the set of $n$ nodes,
$E\subset V\times V$ is the set of edges, and $w: E\rightarrow
\mathbb{R}_{\geq 0}$ is a weight function. We denote the (weighted) adjacency matrix of $G$ by $A$.
For a subset $E'\subset E$ of the edges, we let $w(E') = \sum_{e\in 
E'}w(e)$. Considered $u\in V$, its (weighted) degree is $\dg_u = 
\sum_{e\cap \{v\}\ne\emptyset}w(e)$.\footnote{The term {\em volume} is 
often used rather than weighted degree. Here we simply use the term 
"degree" liberally, since our algorithms and results equally apply to 
unweighted and weighted graphs.} We also let $\dg_{\max} = 
\max_u\dg_u$. Considered $S\subseteq V$, we denote by $G_S$ the induced 
subgraph. The {\em density} $D_S(G)$ of $S\subseteq V$ 
is simply the average degree of $G_S$, namely:
namely:
\[
	D_S(G) = \frac{2\cdot \left\vert E\cap S\times S\right\vert}{\left\vert 
	S\right\vert}.\footnote{The right-hand side becomes $\frac{2w(E\cap 
	S\times S)}{\left\vert S\right\vert}$ for weighted graphs.}
\]
We omit $G$ from $D_S(G)$, whenever clear from context.

A {\em coloring} of the vertices is simply a map $c:V\rightarrow [\ell]$ of $V$,
where $[\ell]:=\{1,2,\dots,\ell$\}. A set $S\subset V$ is called
\emph{fair} if $|S\cap \{v\in V~|~c(v)=1\}| = |S\cap \{v\in
V~|~c(v)=2\}|=\dots= |S\cap \{v\in V~|~c(v)=\ell\}|$. A graph is
called fair if $V$ is fair. 
In the remainder, we provide positive results for the important case 
$\ell = 2$. In this case, for simplicity of exposition we denote the 
colors {\em red} and {\em blue} and we use
$\textit{Red}:=\{v\in V~|~c(v)=\textit{red}\}$ and
$\textit{Blue}:=\{v\in V~|~c(v)=\textit{blue}\}$ to refer
to nodes of the respective color.


\begin{definition}[Fair Densest Subgraph Problem]\label{def:dsg}
Given a (weighted) graph $G(V,E,w)$ and a coloring $c$ of its vertices,
identify a fair subset $S\subseteq V$ that maximizes $D_S$.
\end{definition}

The fair densest subgraph problem is obviously a constrained 
version of the densest subgraph problem. It turns out to be considerably harder 
than its (polynomially solvable) unconstrained counterpart, as we show 
in Section \ref{sec:problems}.

\paragraph{Linear algebra notation.}
We denote by $\lambda_1\ge\lambda_2\ge\ldots\ge \lambda_n$ the 
eigenvalues of $A$ and by $v_i$ its $i$-th eigenvector. We also set 
$\lambda = \max_{i > 2}(\lambda_2, |\lambda_n|)$.
Note that we always have $\lambda_1\le\dg_{\max}$. 
For a subset $S\subset V$, we denote by $\chi$ its 
normalized indicator vector, where $S$ is understood from context. 
Namely, $\chi_i = 1/\sqrt{|S|}$ if $i\in S$, $\chi_i = 0$ otherwise. 
Finally, for a vector $x\in\mathbb{R}^n$, we let $\|x\| =\sqrt{\sum_{i=1}^n
x_i^2}$, the $2$-norm of $x$.

\section{Spectral Relaxations for the Fair Densest Subgraph}\label{se:transform}
As observed in Kannan and Vinay~\cite{KV99}, the densest subgraph 
problem admits a spectral formulation. In particular, denoted by $x$ 
and indicator vector over the vertex set, the indicator vector of the 
vertex subset maxizing density is the maximizer of the following 
expression:
\[
	\max_{x\in\{0, 1\}^n}\frac{x^TAx}{x^Tx}.
\]

Now, assume that each node is colored with one of two colors, red or blue. The optimal solution
$x^*$ might well overrepresent one of the colors.
To formulate the problem of computing a fair solution, we can add the
constraint
\begin{eqnarray*}
& &\sum_{\text{node~$i$ is red}}x_i = \sum_{\text{node~$i$ is blue}}x_i\\
&\Leftrightarrow & \sum_{\text{node~$i$ is red}}x_i\ - \sum_{\text{node~$i$ is blue}}x_i = 0.
\end{eqnarray*}
\[\]
If we define the (unit $2$-norm) vector 
\[
	f_i=
	\begin{cases}
		\frac{1}{\sqrt{n}} 
		&\text{if node~$i$ is red} \\
		-\frac{1}{\sqrt{n}} & \text{if node~$i$ is blue,}
	\end{cases}
\]
the above constraint can be described as
$f^Tx=0$. We call such an $x$ \emph{fair}.
Conversely, very unbiased solutions will have high inner products with $f$.

\paragraph{Fair Densest Subgraph: Spectral Relaxation}
Based on the considerations above, our approach transforms the input data (in this case
the adjacency matrix $A$) by first projecting them onto the kernel of $f$.
Namely, we first consider the following formulation of the fair densest 
subgraph problem:
\[
	\max_{x\in\{0, 1\}^n}\frac{2x^T(I-ff^T)A(I-ff^T)x}{x^Tx}.
\]
It should be noted that, for any fair subset $S$ with indicator $x$, 
we have $\frac{2x^TAx}{x^Tx} = \frac{2x^T(I-ff^T)A(I-ff^T)x}{x^Tx}.$ Conversely, for 
any indicator vector $x\notin span(I-ff^T)$, the objective 
value can only decrease.

We next note that by relaxing $x$ to be an arbitrary vector, the above 
expression is maximized by the main eigenvector of $(I-ff^T)A(I-ff^T)$. Indeed,~\cite{KV99} established a relationship between the first 
eigenvector of the adjacency matrix $A$ and an approximately densest subgraph. 
Similar ideas are also implicit in the work of~\cite{McS01}. 
The above relaxation corresponds to replacing hard fairness constraints 
with soft ones. To prove our main result we need the following 
definition:

\begin{definition}\label{def:ar}
Graph $H = (V_H, E_H)$ is $(\dg, \epsilon)$-regular if a $\dg$ exists, 
such that $(1 - \epsilon)\dg\le\dg_i\le(1 + \epsilon)\dg$, for every 
$i\in V_H$.
\end{definition}

\begin{theorem}\label{thm:fair_spectral}
Assume we have a graph $G = (V, E, w)$ with a 2-coloring of the nodes.
Assume the spectrum of $A$ satisfies $\lambda_1\ge 
4\lambda$.\footnote{That is, $G$ is an expander.} Assume further that $G$ 
contains a fair subset $S$ such that: (1) $G_S$ is $(\dg, 
\epsilon)$-regular and (2) $d\ge(1 - \theta)\dg_{\max}$. In this case, it is 
possible to recover all but $16(\epsilon + 
\theta)|S|$ of the vertices in $S$ in polynomial time.
\end{theorem}
Intuitively, the result above states that, if the underlying network 
$G$ is an expander containing an almost-regular, dense and fair subgraph, we 
can approximately retrieve it in polynomial time. Succintly, this 
follows because, under these assumptions, the 
indicator vector of $S$ forms a small angle with the main eigenvector of 
$(I - ff^T)A(I - ff^T)$. 

\begin{proof}[Proof of Theorem \ref{thm:fair_spectral}]
In the remainder of this proof, we denote by $\hlambda_1\ge,\ldots 
,\ge\hlambda_n$ the eigenvalues of $(I - ff^T)A(I - ff^T)$ and by $\hv_i$ the 
$i$-th associated eigenvector. For a vertex $i$ of $G_S$ we denote by 
$\hat{\dg}_i$ its degree in $G_S$. We denote by $\chi$ the indicator 
vector of $S$ and we let $m = |S|$.

As a first step, we summarize straightforward, yet useful properties of the 
spectrum of $(I - ff^T)A(I - ff^T)$.
\begin{lemma}\label{le:hvs}
Whenever $\hlambda_i\ne 0$ we have:
\begin{equation}
	(I - ff^T)\hv_i = \hv_i\, \text{and }\,\hlambda_i = \hv_i^TA\hv_i
\end{equation}
\end{lemma}
\begin{proof}
If $\hlambda_i\ne 0$, we have:
\[
	(I - ff^T)A(I - ff^T)\hv_i = \hlambda_i\hv_i.
\]
Since $(I - ff^T)$ is a projection matrix, if we pre-multiply both 
members of the above equation by $(I - ff^T)$ we have:
\[
	(I - ff^T)A(I - ff^T)\hv_i = \hlambda_i(I - ff^T)\hv_i.
\]
Subtracting the first equation from the second and recalling that 
$\hlambda_i\ne 0$ immediately the first claim.

The second claim follows immediately from the first:
\[
	\hlambda_i = \hv_i^T(I - ff^T)A(I - ff^T)\hv_i = \hv_i^TA\hv_i.
\]
\end{proof}
It should be noted that, as a consequence of Lemma \ref{le:hvs}, we 
always have:
\[
	\hlambda_1 = \hv_1^T(I - ff^T)A(I - ff^T)\hv_1 = 
	\hv_1^TA\hv_1\le v_1^TAv_1 = \lambda_1.
\]
Note that this last property does not apply to the other eigenvalues in general.
The first important, technical step to prove Theorem \ref{thm:fair_spectral} is showing that the 
hypothesis $\lambda_1\ge 4\lambda$ implies that $\hat{\lambda}_2$ 
cannot be ``too large''.
\begin{lemma}\label{le:gap}
Assume the spectrum of $A$ satisfies the condition $\lambda_1\ge 
4\lambda_2$. Then $\hlambda_2\le\frac{3}{4}\lambda_1$.
\end{lemma}
\begin{proof}
We first express $\hv_2$ as $\hv_2 = \gamma v_1 + z$, where $z$ is 
$\hv_2$'s component orthogonal to $v_1$, the main eigenvector of $A$. 
Note that, since $v_1$ has unit norm, we have $\gamma^2 + \|z\|^2 = 1$. 
Next:
\begin{equation}\label{eq:hlambda_2}
	\hlambda_2 = (\gamma v_1 + z)^TA(\gamma v_1 + z)^T = \gamma\lambda_1 + 
	z^TAz,
\end{equation}
where the first equality follows from Lemma \ref{le:hvs}, while the 
second follows since $z\in \mathbf{span}(v_2,\ldots , v_n)$ by definition and 
the $v_i$'s form an orthonormal basis.
Next, assume $\gamma\ge 1/2$. In this case, we have:
\begin{align}\label{eq:contr}
	&\hlambda_2 = \gamma\lambda_1 + z^TAz = \gamma\lambda_1 + 
	\|z\|^2\frac{z^TAz}{\|z\|^2}\ge \frac{3}{8}\lambda_1.
\end{align}
Here, the third inequality follows since i) $\|z\|^2 = 1 - \gamma^2\le 
1/2$, while $z\in \mathbf{span}(v_2,\ldots , v_n)$ implies:
\[
\left|\frac{z^TAz}{\|z\|^2}\right|\le\max_{w\perp 
v_1}\left|\frac{w^TAw}{\|w\|^2}\right| = \lambda.
\]
But \eqref{eq:contr} contradicts our assumption that $\lambda_1\ge 
4\lambda$. On the other hand, if $\gamma\le 1/2$, \eqref{eq:hlambda_2} implies:
\begin{align*}
	&\hlambda_2 = \gamma\lambda_1 + z^TAz\le\gamma\lambda_1 + 
	z^TAz\\
	&\le\gamma\lambda_1 + \|z\|^2\max_{w\perp 
	v_1}\frac{w^TAw}{\|w\|^2}\\
	&= \gamma\lambda_1 + (1 - \gamma^2)\lambda_2\le\frac{\lambda_1}{2} + 
	\lambda_2\le\frac{3}{4}\lambda_1. 
\end{align*}
\end{proof}
The second step is showing that Lemma \ref{le:gap} implies that the 
indicator vector of the fair densest subgraph is close to $\hv_1$:
\begin{lemma}\label{le:projection}
Assume the the hypotheses of Theorem \ref{thm:fair_spectral} hold. Then:
\[
	\|\chi - \hv_1\|^2\le 4(\epsilon + \theta).
\] 
\end{lemma}
\begin{proof}
We begin by noting that $\chi^Tf = 0$ by definition, 
which implies $(I - ff^T)\chi = \chi$. 
We therefore have:
\begin{equation}
	\chi^T(I - ff^T)A(I - ff^T)\chi = \frac{\sum_{i\in 
	S}\hdg_i}{m}\ge (1 - \epsilon)d,\label{eq:1}
\end{equation}
Next, we decompose $\chi$ along its components respectively parallel 
and orthogonal to $\hv_1$, namely, $\chi = \alpha\hv_1 + z$, and we 
note that $\|z\|^2 = 1 - \alpha^2$, since both $\hv_1$ and $\chi$ are 
unit norm vectors. Set $B = (I - ff^T)A(I - ff^T)$ for the sake of 
space. We have:
\begin{align}
\nonumber
	&\chi^TB\chi = (\alpha\hv_1 + z)^TBs(\alpha\hv_1 + z)
	= \alpha^2\hlambda_1 + z^TBz \\
	&\le\alpha^2\hlambda_1 + \hlambda_2\|z\|^2\le\alpha^2\hlambda_1 
	+ (1 - \alpha^2)\hlambda_2.\label{eq:2}
\end{align}
Putting together \eqref{eq:1} and \eqref{eq:2} yields
$	\alpha^2\ge\frac{(1 - \epsilon)\dg - \hlambda_2}{\hlambda_1 - 
	\hlambda_2}$. Now:
\begin{align*}
	& \|\chi - v\|^2\le 1 - \frac{(1 - \epsilon)\dg - \hlambda_2}{\hlambda_1 
	- \hlambda_2}\\
	&\le 1 - \frac{(1 - \epsilon)(1 - \theta)\dg_{\max} - \hlambda_2}{\lambda_1 
	- \hlambda_2}\\
	&\le 1 - \frac{(1 - \epsilon)(1 - \theta)\lambda_1 - \hlambda_2}{\lambda_1 
	- \hlambda_2} \\
	&< 1 - \frac{\lambda_1 - \hlambda_2 - (\epsilon + \theta)\lambda_1}{\lambda_1 
	- \hlambda_2} \\
	&= \frac{(\epsilon + \theta)\lambda_1}{\lambda_1 - 
	\hlambda_2}\le 4(\epsilon + \theta).
\end{align*}
Here, the second inequality follows from our hypotheses on $\dg$ and 
since $\hlambda_1\le\hlambda$, the 
third inequality follows since the main eigenvalue of an adjacency 
matrix is upper-bounded by the maximum degree of the underlying graph, 
while the last inequality follows from Lemma \ref{le:gap}.
\end{proof}

\begin{corollary}\label{cor:no_outliers}
Under the hypotheses of Lemma \ref{le:projection}, for all but at most 
$16m(\epsilon + \theta)$ vertices in $V$ we have: i) $\hv_1(i) \ge 
\frac{1}{2\sqrt{m}}$ if $i\in S$, ii)  $\hv_1(i) < \frac{1}{2\sqrt{m}}$ 
otherwise.
\end{corollary}
\begin{proof}
We first note that the number of vertices $i$ for which 
$	\left |\hv_1(i) - \chi(i)\right | > \frac{1}{2\sqrt{m}}$
is at most $16m(\epsilon + \theta)$. To see why, assume there are $x$ such vertices. This implies
$	\|\hv_1 - \chi\|^2\ge\frac{x}{4m}.$
On the other hand, the upper bound following from Lemma 
\ref{le:projection} immediately implies $x\le 16m(\epsilon + \theta)$.
Next, recall that $\chi_i = \frac{1}{\sqrt{m}}$ if $i\in S$, $\chi_i = 
0$ otherwise. This immediately implies the thesis.
\end{proof}
\paragraph{The algorithm}
Our algorithm is based on a sweep of $\hv_1$ \cite{KV99,McS01}. In 
particular, we run Algorithm GSA (see Algorithm \ref{alg:ss}) with $M = (I - ff^T)A(I - 
ff^T)$ and $\Delta = 16(\epsilon + \theta)$.
\begin{algorithm}[h]
\TitleOfAlgo{General Sweep Algorithm (GSA)}
\SetAlgoLined
\KwData{Non-negative $n\times n$ matrix $M$, parameter $\Delta$}
\KwResult{Subset $S\subseteq V$}
	$\hat{S} = \emptyset$; $\hat{D} = 0$\;
	
	Compute $v_1 = $ main eigenvector of $M$\;
	
	Sort nodes $i\in V$ in non increasing order of 
	$v_1(i)$\;
	
	\tcp*[h]{Assume w.l.o.g. that $\{1, \ldots, n\}$ is resulting ordering of 
	nodes in $V$}\;
	
	\For {$s = 1$ to $n$} {
		$S = \{1,\ldots , s\}$\\
		Compute $D_S = $ density of the subgraph induced by $S$\\
		\If {$D_S > \hat{D}$ AND $\left ||S\cap Red| - |S\cap Blue|\right | 
		\le \Delta |S|$} {
			$\hat{S} = S$; $\hat{D} = D_S$\\
		}
	}
	return $\hat{S}$
\caption{General Sweep Algorithm (non-increasing).}
\label{alg:ss}
\end{algorithm}

Corollary \ref{cor:no_outliers} ensures that i) the above algorithm always 
returns a solution, ii) the solution returned by the algorithm will not 
be worse than the one obtained by picking $i$ if $\hv_1(i)\ge 
\frac{1}{2\sqrt{m}}$ and rejecting it otherwise. 
This concludes the proof of Theorem \ref{thm:fair_spectral}.
\end{proof}

\section{Hard Constraints and Hardness of Approximation}
\label{sec:problems}
In general, enforcing fairness can make an ``easy'' problem intractable
and this is the case for the densest subgraph problem.
In this context, spectral relaxations can be regarded as a way to mitigate this issue, by 
enforcing soft fairness constraints to virtually any problem that is
amenable to an algebraic formulation. 

Nevertheless, in some cases it might be important to assess 
the \emph{price of fairness}, by comparing the achievable quality of fair solutions 
to that of solutions for the original, unconstrained problem.
In this section, we complement our algorithmic treatment of fairness with 
hardness results and approximation algorithms for specific cases. 
Some of our hardness results are based on the \emph{small set expansion
hypothesis}, which we now describe.

Consider a $d$-regular weighted graph $G$ and, for every $S\subset V$, 
denote by $\Phi(S)$ the \emph{expansion}\footnote{Or conductance.} of $S$ \cite{RaS10}. 
Given two constants $\delta,\eta\in
(0,1)$, the small set expansion problem~\cite{RaS10}
$SSE(\delta,\eta)$ asks to distinguish between the following two cases:
\begin{description}
\item[Completeness] There exists a set of nodes $S\subset V$ of size
$\delta\cdot |V|$ such that $\Phi(S)\leq \eta$.
\item[Soundness] For every set of nodes  $S\subset V$ of size
$\delta\cdot |V|$, $\Phi(S)\geq 1-\eta$.
\end{description} 
Our hardness proofs are based on the small set
expansion hypothesis defined as follows.
\begin{conjecture}[SSEH]
For every $\eta>0$ there exists a $\delta:=\delta(\eta)>0$ such that
$SSE(\eta,\delta)$ is NP-hard.
\end{conjecture}

\subsection{General Case}
Given a weighted graph $G(V,E,w)$, the densest subgraph problem asks to
find a set of nodes $S$ such that the density $D_S:=\frac{|E_S|}{|S|}$ is maximized.
The fair densest subgraph problem additionally requires $S$ to be fair.

The densest at-most-$k$ subgraph problem asks to find a set of nodes $S$
with $|S|\le k$, such that the density $D_S:=\frac{|E_S|}{|S|}$ is
maximized.
Recall from Section~\ref{sec:related} that, whereas the densest subgraph
problem is polynomially solvable, the best approximation for the densest
at-most-$k$ subgraph problem is in $O(n^{1/4})$~\cite{BhaskaraCVGZ12} and cannot be approximated up to a factor of
$n^{1/(\log\log n)^c}$ for some $c>0$ assuming the exponential time hypothesis~\cite{Manurangsi17} . The next
theorem implies that these inapproximability results for the densest
at-most-$k$ subgraph problem hold also for the fair densest subgraph
problem, showing that fairness constraints can drastically affect 
hardness of this problem.

%


\begin{theorem}
Computing an $\alpha$-approximation for the fair densest subgraph
problem is at least as hard as computing an $\alpha$-approximation for
the densest at-most-$k$ subgraph problem.
Moreover, any $\alpha$-approximation to the densest at-most-$k$ subgraph is an $2\alpha$ approximation to densest fair subgraph.
\end{theorem}
\begin{proof}
Consider an arbitrary graph $G(V,E)$. We consider $V$ to be colored red.
Add $k$ blue nodes with no edges. Then the density of the fair densest
subgraph is, up to a multiplicative factor of exactly $\frac{1}{2}$, equal to
the density of the densest at most $2 k$ subgraph.

The argument for the upper bound is completely analogous to the proof given for Theorem~\ref{thm:approxalg}.
\end{proof}

When the input graph $G$ is itself fair, we can provide stronger bounds.

\begin{algorithm}
\caption{Approximate Fair Densest Subgraph}
\label{alg:adfsg}
\let\oldnl\nl
{\bf Input:} Graph $G(V,E,w)$\\
1: Compute the densest subgraph $S$\\
2: W.l.o.g $|S\cap Blue|\geq |S\cap Red|$\\ 
3: While $|S\cap Blue| > |S\cap Red|$, add an arbitrary node $v\in
Red\setminus S$ to $S$\\
4: Return $S$
\end{algorithm}

\begin{theorem}
\label{thm:approxalg}
Given a fair graph $G(V,E,w)$, Algorithm~\ref{alg:adfsg} computes a fair
set $S\subset V$, such that $D_S\cdot 2\geq OPT$, where $OPT$ is the
density of the fair densest subgraph.
Similarly, any $\alpha$-approximation to densest at-most-$k$ subgraph is an $2\alpha$ approximation to densest fair subgraph.
\end{theorem}
\begin{proof}
We refer to the set $S$ computed after line $1$, and $3$ as $S_1$ an
$S_2$, respectively.
Since $S_1$ is the unconstrained densest subgraph, $D_{S_1}>OPT$. For
$S_2$, we observe that $|S_2| \leq S_1 + |S_1\cap Blue| - |S_1\cap Red|
\leq 2\cdot |S_1|$, hence 
\[D_{S_2} = \frac{w(E_{S_2})}{|S_2|} \geq \frac{w(E_{S_1})}{2|S_1|} \geq
\frac{OPT}{2}.\]
%
To conclude the proof, we observe that $S_2$ is always fair.
\end{proof}

We conclude this section by showing that approximating the fair densest
subgraph problem beyond a factor of $2$ is at least as hard as solving
$SSE(\eta,\delta)$. Therefore, barring a major algorithmic breakthrough,
Algorithm~\ref{alg:adfsg} is optimal. The proof is provided as
supplementary material and it is based on the following idea: 
In regular graphs, for a given set of nodes $S$, the
expansion $\Phi(S)$ is related to the density of~$S$. We can use this,
so that, given a graph $G$, we can carefully construct a colored graph
$G'$ such that finding the optimal fair densest subgraph in $G'$ gives
an estimate of the largest-expansion node set in $G$.

\begin{theorem}
\label{thm:hard-densest-sseh}
If SSEH holds, computing a $(2-\varepsilon)$ approximation of the
fair densest subgraph problem in fair graphs is $NP$-hard for any
$\varepsilon>0$.
\end{theorem}
\begin{proof}
We consider the $SSE(\eta,\delta)$ problem, i.e. let $G(V,E,w)$ be a
$d$-regular graph and let $\eta\in (0,1)$ and $\delta = \delta(\eta)\in
(0,1/2]$ be constants that we will specify later.
For any set $S\subset V$ of size $s:=\delta\cdot |V|$, we have $w(E_S)
:= d\cdot s - \Phi(S)\cdot d \cdot s$.

We construct a colored graph $G'(V',E',w')$ by considering all nodes of
$G$ to be colored red, and by adding $|V|$ blue nodes. Of these nodes,
we select an arbitrary but fixed subset of $\delta\cdot |V|$ blue nodes
that we denote by $B$. Each edge in $E_B$ is weighted uniformly by
$t:=\frac{2\cdot d}{s-1}$. The remaining edges are weighted with $0$.

Recall that SSEH states that distinguishing between the two cases is
$NP$-hard.

\paragraph{Completeness} 
If there exists some $S\subset V$ of size $s$ with $\Phi(S)\leq \eta$, then 
\begin{equation*}
\label{eq:case1}
w(E_S) \geq (1-\eta)\cdot d\cdot s.
\end{equation*} 
Then the density of the fair subgraph induced by $S\cup B$ of size $2s$
satisfies
\begin{eqnarray}
\nonumber
D_{S\cup B} &=& \frac{w(E_S) + w(B)}{2|S|} \geq \frac{(1-\eta)\cdot d\cdot s + t\cdot {s \choose 2}}{2s}  \\
\label{eq:upper}
&=& \frac{(1-\eta)\cdot d + t \cdot \frac{s-1}{2}}{2} \geq (1-\eta)\cdot d.
\end{eqnarray}

\paragraph{Soundness} If for all $S\subset V$ of size $s$, $\Phi(S)\geq 1-\eta$, then 
\begin{equation}
\label{eq:case2}
w(E_S) \leq \eta \cdot d\cdot s.
\end{equation} 

Denote the size of the fair densest subgraph $C$ by $k$. Further, let
$C_{red}=C\cap Red$.
We will distinguish between four basic cases: (1) $k<2 \mu\cdot s$, (2)
$2 \mu\cdot s \leq k < 2\cdot s$, (3) $2\cdot s \leq k <\frac{2}{\mu}s$,
and (4) $\frac{2}{\mu}s \leq k$, where $\mu>0$ is suitably small
constant specified later. 
We note that the cases (1) and (4) and (2) and (3) will turn out to be
somewhat symmetric, even if slightly different proofs are required in
every case.


First, let $k< 2 \mu\cdot s$ and again let $B_{k}$ be an arbitrary
subset of $B$ of size $k$. Then

\begin{eqnarray}
\nonumber
D_{C_{red}\cup B_{k}} &\leq & \frac{d\cdot k + w(B_{k})}{2\cdot k} = \frac{d\cdot k + t\cdot {k \choose 2}}{2\cdot k}  \\
\nonumber
&=& \frac{d + t\cdot \frac{k-1}{2}}{2} = \frac{d}{2} + \frac{d\cdot (k-1)}{2(s-1)} \\
\label{eq:lower2}
&\leq &  \frac{d}{2} + \frac{d\cdot 2\mu}{2} \leq (1+2\mu)\frac{d}{2},
\end{eqnarray}
where the first inequality holds due to regularity.

Now, let $2 \mu\cdot s\leq k < 2 \cdot s $.
We have 
\begin{eqnarray}
\nonumber
D_{C_{red}\cup B_{k}} &\leq & \frac{\eta\cdot d\cdot s + w(B_{k})}{2\cdot k}  \\
\nonumber &=& \frac{\eta\cdot d\cdot s}{2\cdot k} + \frac{d\cdot (k-1)}{2\cdot(s-1)} \\
\label{eq:lower1}
&\leq &  \frac{\eta\cdot d}{\mu} + \frac{d}{2} \leq \left(1+\frac{2\eta}{\mu}\right)\frac{d}{2}.
\end{eqnarray}

Now, let $2\cdot s\leq k\leq \frac{2}{\mu}\cdot s$.
We will first show that 
\begin{equation}
\label{eq:counting}
w(C)\leq \frac{2}{\mu}\cdot \eta\cdot d\cdot k.
\end{equation}
For the sake of contradiction, assume that this is not the case.
The argument revolves around double counting $w(C)$.
There exist ${k \choose s}$ subsets of size $s$ of $C$. 
Observe that for any such subset $S'$ has weight $w(S')\leq \eta\cdot d\cdot s$ and hence
\[\sum_{S'\subset C~\wedge ~ |S'|=s} w(S') \leq \eta\cdot d\cdot s\cdot {k\choose s}.\]
At the same time, every (possibly $0$ valued) edge appears in ${k-2\choose s-2}$ of these subsets. Hence
\[\sum_{S'\subset C~\wedge ~ |S'|=s} w(S') = w(C)\cdot {k-2\choose s-2}> \frac{2}{\mu}\cdot \eta\cdot d\cdot k \cdot {k-2\choose s-2}.\]
Combining both equations, we have
\begin{align*}
\frac{2}{\mu}\cdot \eta\cdot d\cdot k \cdot {k-2\choose s-2} < \eta\cdot d\cdot s\cdot {k\choose s} \\
\Leftrightarrow
\frac{2}{\mu} < \frac{k\cdot(k-1)}{s\cdot(s-1)}\frac{s}{k} \leq \frac{2}{\mu},
\end{align*}
which is a contradiction.

Consider now the density of any fair cut containing $C\cup B_{k}$, where
$B_{k}$ contains $B$ and $k-s$ further arbitrary blue nodes. We have
\begin{eqnarray}
\nonumber
D_{C_{red}\cup B_{k}} &\leq & \frac{\frac{2\eta}{\mu}\cdot d\cdot k + t\cdot {s\choose 2}}{2\cdot k} \\
\label{eq:lower3}
&=&  \frac{\eta}{\mu}\cdot \frac{d}{2} + \frac{d\cdot s}{2 \cdot k} \leq  \left(1+\frac{2\eta}{\mu}\right)\cdot \frac{d}{2}.
\end{eqnarray}

Finally, consider the case $k>\frac{2}{\mu}s$.
Then the density of any fair cut containing $C\cup B_k$, where $B_{k}$
contains $B$ and $k-s$ further arbitrary blue nodes, is
\begin{eqnarray}
\nonumber
D_{C_{red}\cup B_{k}} &\leq & \frac{d\cdot k + t\cdot {s\choose 2}}{2\cdot k} \\
\label{eq:lower4}
&=&  \frac{d}{2} + \frac{d\cdot s}{2\cdot k} \leq  \leq  (1+2\mu)\frac{d}{2}.
\end{eqnarray}

We note that bounds from Equations~\ref{eq:lower2} and~\ref{eq:lower3}
and Equations~\ref{eq:lower1} and~\ref{eq:lower4} are identical.
For $\varepsilon<\frac{1}{4}$, we set $\mu = \frac{\varepsilon}{2}$,
$\eta\leq \frac{8}{3}\cdot\varepsilon^2$. Then the ratio between the
terms~\ref{eq:upper} and~\ref{eq:lower2} and the terms \ref{eq:upper}
and~\ref{eq:lower1} is at least $2-\varepsilon$.  Therefore,
approximating the fair densest subgraph problem beyond a factor of $2$
solves the $SSE(\eta,\delta)$ problem.
\end{proof}

 \section{Experimental Analysis}
\label{se:experiments}

Worst case bounds are often uninformative when compared with empirical behavior. Algorithm~\ref{alg:adfsg} is (assuming that the underlying graph is fair) theoretically optimal and therefore superior to the spectral recovery schemes. As we now describe, the empirical performance between these approaches paints the opposite picture.

\paragraph{Overview.}
To test the performances of our algorithms on real data we used two publicly available dataset: 
\PolBooksData~\cite{PolBookNetwork} and \AmazonData products 
metadata~\cite{ni-etal-2019-justifying}. 
Both (explicitly or implicitly) contain undirected 
unweighted graphs, whose nodes are products from the Amazon catalog, 
while an edge between two nodes exists if the corresponding products 
are frequently co-purchased by the same buyer. Moreover, for both datasets, each product belongs 
to exactly one category. 

We tested our methods in a scenario in which, given a 
(not necessary fair) labelled graph, 
our only interest lies in finding fair subgraphs with high density.
In this context, we are considering the density of the provided solution as 
a quality indicator: more dense reflects more quality.

For our experiments we used an Intel Xeon 2.4GHz 
with 24GB of RAM running Linux Ubuntu 18.04 LTS.

\paragraph{Datasets.} 
The \textsc{PolBooks} data set~\cite{PolBookNetwork} is an undirected 
unweighted graph\footnote{http://www.casos.cs.cmu.edu/computational\_tools/ datasets/external/polbooks/polbooks.gml.}, 
whose nodes represent books on US politics included 
in the Amazon catalog, while an edge between two books exists if both 
books are frequently co-purchased by the same buyer. Each book is 
further labeled depending on its political stance, possible labels 
being ``\textit{liberal}'', ``\textit{neutral}'', and ``\textit{conservative}''. 
For our experiments, we considered only the subgraph 
induced by ``\textit{liberal}'' and ``\textit{conservative}'' books, obtaining 92 nodes 
(49 of which were associated with a ``\textit{conservative}'' worldview, 43 with 
a ``\textit{liberal}'' worldview) for 362 edges in total.

The \AmazonData products metadata 
dataset~\cite{ni-etal-2019-justifying} contains descriptions for 15.5 
million Amazon 
products~\footnote{https://nijianmo.github.io/amazon/index.html}. For
a single product, we only considered the product id (\textit{asin} field), the category 
the product belongs to (\textit{main\_cat} field) and 
the set of frequently co-purchased products (\textit{also\_buy} field). It should be noted 
that in this dataset, each node belongs to exactly one (main) Amazon 
category so that, together, these three fields allow recovery of a 
large, undirected, labelled graph, with products as nodes, 
categories as labels and edges representing frequent co-purchasing product pairs.
For this data set, we leveraged the co-purchasing relation among 
products, to naturally extract undirected and unweighted, labelled 
graphs. In more detail, for each pair $(\ell_1, \ell_2)$ of Amazon main categories, 
we extracted the undirected subgraph induced by the subset of 
nodes of category $\ell_1$ ($\ell_2$) that have at least one neighbour from 
category $\ell_2$ ($\ell_1$). We did not consider graphs with fewer than 100 nodes. This 
way, we retrieved 292 subgraphs, with sizes ranging between 100 and 22046 nodes.

\paragraph{Algorithms.} We compared the performance of the following algorithms:

\noindent{{\bf 2-DFSG.}} The optimal 2-approximation 
algorithm (Algorithm~\ref{alg:adfsg}) based on Goldberg's optimal algorithm for the densest 
subgraph problem \cite{G84}, described in Section \ref{sec:problems}.

\medskip\noindent{\bf Spectral Algorithms.}
Following~\cite{KV99,McS01} and 
Theorem~\ref{thm:fair_spectral}, we ran a variety of eigenvector 
rounding algorithms. 
These are all
variants of 
a modified version of the General Sweep Algorithm (Algorithm~\ref{alg:ss}) used 
in the proof of Theorem \ref{thm:fair_spectral} 
that sorts the entries of the main eigenvector of $M$ 
four times (instead of a single one)
according to the following criteria: 
i) non-increasing; ii) non-decreasing; 
iii) non-increasing absolute values; iv.) non decreasing absolute values.
With these premises, we consider the 
following spectral algorithms. 
The first two are just the modified version of Algorithm~\ref{alg:ss} 
with different choices for $M$, while {\bf PS} and {\bf 
FPS} perform a slightly modified sweep that always affords a fair 
solution.

\noindent{\bf Single Sweep (SS).} 
This algorithm is simply (Algorithm~\ref{alg:ss}),
when all previously mentioned sorting criteria are used, 
with $M = A$ and $\Delta = 0$.

\noindent{{\bf Fair Single Sweep (FSS).} 
It is the execution of SS, 
this time on matrix $(I-ff^T)A(I-ff^T)$ instead of $A$.

\noindent{{\bf Paired Sweep (PS).}} Paired Sweep is a modification of  
SS in which the fairness constraint is satisfied by 
construction in each subgraph produced by the rounding algorithm. This is 
done by considering the subsets $V_{\textit{Red}}$ and 
$V_{\textit{Blue}}$ of the nodes, sorting each of them separately 
according to the values of the corresponding entries in the main 
eigenvector of $A$ and then, for each $s = 1,\ldots, 
\min({|V_{\textit{Red}}|, |V_{\textit{Blue}}|})$ considering the 
candidate set of nodes of cardinality $2s$ obtained by taking the first $s$ 
nodes from each ordered subset.

\noindent{{\bf Fair Paired Sweep (FPS).} It is the execution of PS, 
this time on matrix $(I-ff^T)A(I-ff^T)$ instead of $A$.

\subsection{Results}
\label{sec:experiments}

\begin{figure}[h!]
\begin{center}
\includegraphics[scale=0.4]{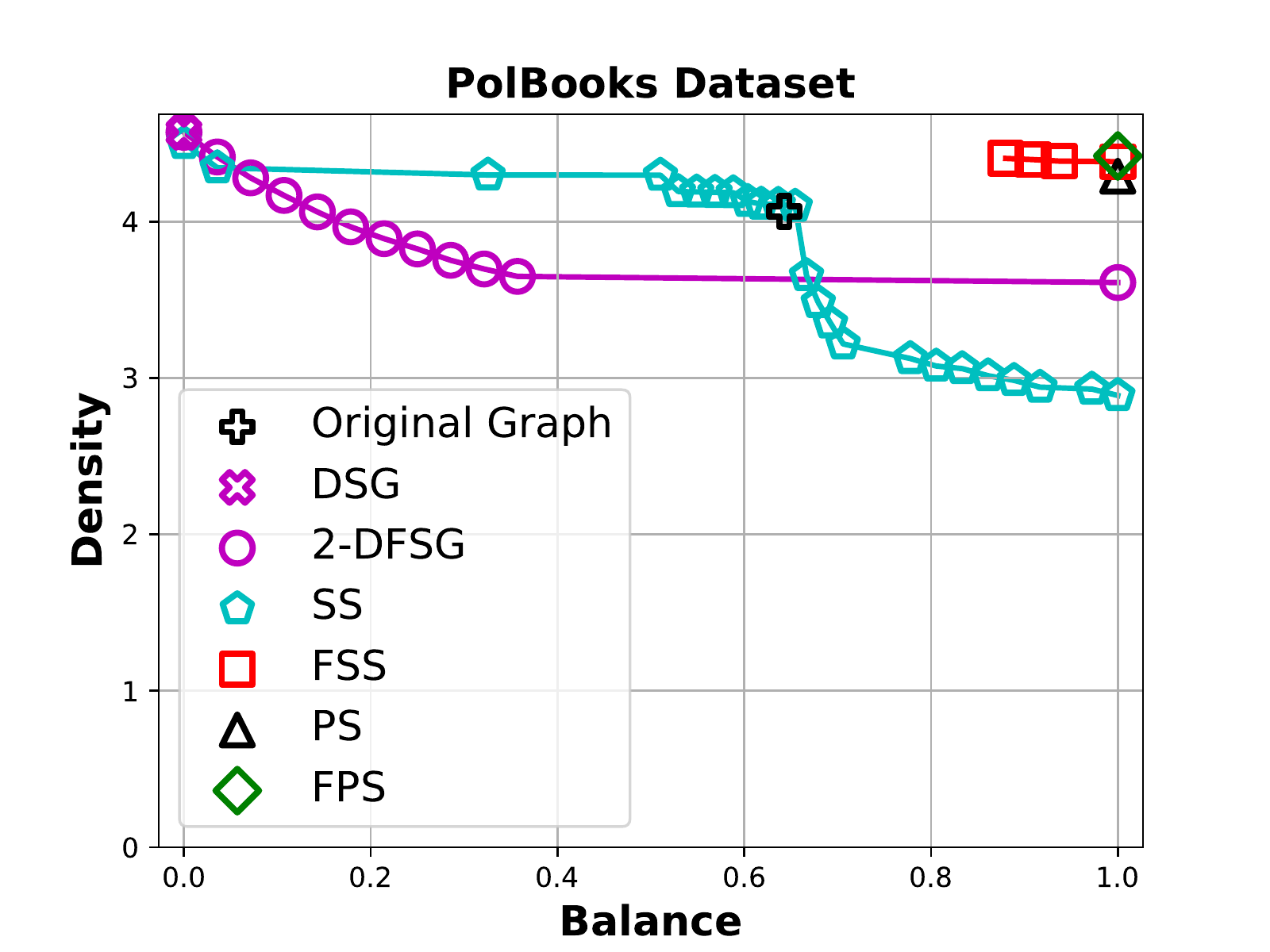}
\end{center}
\caption{Pareto front of the subgraphs generated by each algorithm,
w.r.t. density and balance, 
on \PolBooksData dataset.}
\label{fig:dfsg_exp_0}
\end{figure}

\begin{figure}[h!]
\begin{center}
\includegraphics[scale=0.4]{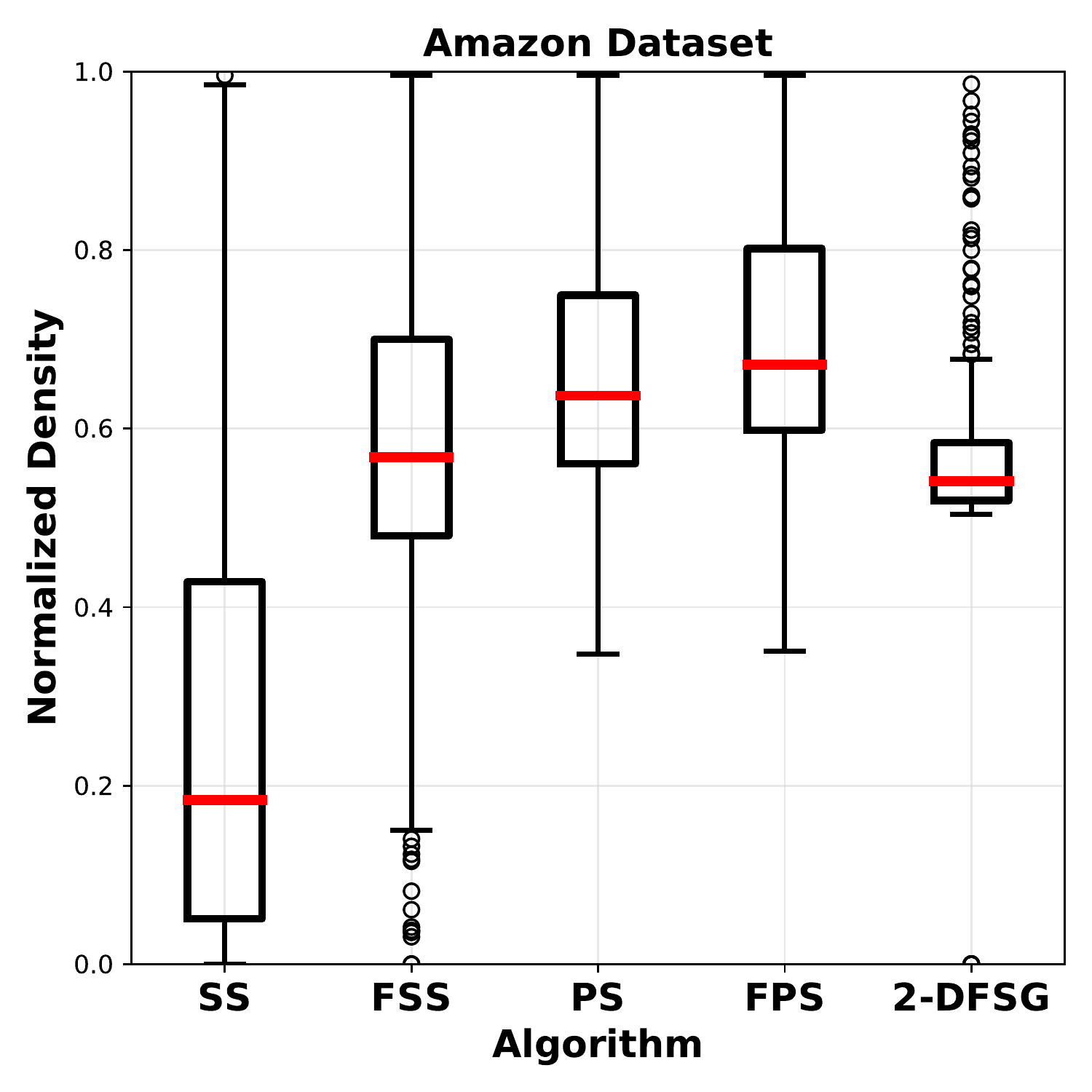}
\end{center}
\caption{Performance of our algorithms on \AmazonData dataset on 292 samples. Reported are aggregates over all generated subgraphs, with unfair solutions receiving a density of $0$, see Table~\ref{tab:fair}.}
\label{fig:dfsg_exp_1}
\end{figure}

Figure \ref{fig:dfsg_exp_0} shows the performance of our algorithms on \PolBooksData dataset 
through the Pareto front of the subgraphs generated by each algorithm during its execution w.r.t. density and balance~\footnote{Given two color classes \textit{Red} and \textit{Blue}, we define the 
\textit{balance} of a subgraph containing $x$ \textit{Red} and $y$ 
\textit{Blue} nodes as $\min\left(\frac{x}{y}, \frac{y}{x}\right)$.}. 
PS and FPS by construction only return fair solutions while the other algorithms potentially have trade-offs. In particular, 2-DSG (Algorithm~\ref{alg:adfsg}) starts at the unconstrained optimum and proceeds to add nodes that increase balance while potentially decreasing density.

Figure \ref{fig:dfsg_exp_1} shows the distributions of the normalized 
density, over the entire set of \AmazonData instances, of the fair subgraphs 
retrieved by different algorithms. Normalization, performed
to make solutions for different instances comparable, is done 
by scaling to the optimal density of the unconstrained 
problem.\footnote{Hence, the maximum possible value on the $y$-axis is 
$1$.}
With the exception of SS, which uses the original adjacency matrix and 
whose distribution is skewed toward lower density values, 
performances of spectral heuristics are comparable, with FPS achieving highest 
median density. In general spectral 
algorithms run on $(I-ff^T)A(I-ff^T)$ (FSS and FPS) respectively outperform their counterparts 
(SS and PS) run on $A$.
Interestingly, with the exception of SS, spectral heuristics 
consistently outperform 2-DFSG, despite its theoretical optimality. 

We report in Table \ref{tab:non_fair_solutions} the percentage of 
instances each algorithm is not able to solve, i.e., for which it does 
not return a fair solution and, consequently, 
we assigned a density equal to $0$.

\begin{table}[h!]
\begin{center}
\begin{tabular}{|c|c|c|c|c|}
\hline
\textbf{SS} & \textbf{FSS} & \textbf{PS} & \textbf{FPS} & \textbf{2-DFSG} \\ \hline
1.03 & 0.34 & 0 & 0 & 3.08 \\ \hline
\end{tabular}
\caption{Percentages of unfair solutions over 292 samples for \AmazonData dataset. As noted previously, PS and FPS cannot return unfair solutions. $2$-DFSG (Algorithm~\ref{alg:adfsg}) results in an unfair solutions if the original graph is unbalanced and the unconstrained densest subgraph cannot be made fair via line 3.}
\label{tab:non_fair_solutions}
\end{center}
\end{table}\label{tab:fair}

\section{Conclusion and Future Work}
In this work, we studied graphs with an arbitrary $2$-coloring. For these graphs, the densest fair subgraph problem consists of finding a subgraph with maximal induced degree under the condition that both colors occur equally often. We observed that the problem is closely related to the densest at most k subgraph problem and thus has similar strong inapproximability results.
On the positive side, we presented an optimal approximation algorithm under the assumption that the graph itself is fair, and a more involved spectral recovery algorithm inspired by the work of~\cite{KleindessnerSAM19} on stochastic block models.

In practice, the spectral recovery algorithm tended to dominate the approximation algorithm.
We interpret these results as showing that (1) an approximation algorithm 
may not be the correct way to attack this problem, and (2) 
as previous work also suggests \cite{SamadiTMSV18,KleindessnerSAM19}, spectral 
relaxations seem to be an inexpensive tool to improve the fairness of 
algorithms geared towards recovery and learning. 

Future work might consider extending this approach to more involved fairness constraints. As noted by~\cite{KleindessnerSAM19}, the ideal of removing ``unfairness'' via orthogonal projections straightforwardly generalized to multiple and potentially overlapping color classes. However, analyzing the spectrum in these cases seems to be significantly more difficult.

\newpage
\bibliographystyle{plain}
\bibliography{references}

\end{document}